\documentclass[submission,copyright,creativecommons]{eptcs}
\usepackage[english]{babel}
\usepackage{amssymb}
\usepackage{amsmath}
\usepackage{amsthm}
\usepackage{stmaryrd}
\usepackage[T1]{fontenc}
\usepackage{tikz}
\usetikzlibrary{calc}
\usepackage{array}   %stellt den Befehl \newcolumntype bereit 
\newcolumntype{C}[1]{>{\centering}p{#1}} %neuer Spaltentyp, zentriert + Breitenangabe moeglich
\providecommand{\event}{WORDS 2011} % Name of the event you are submitting to

\newtheorem{theorem}{Theorem}
\newtheorem{lemma}[theorem]{Lemma}
\newtheorem{definition}[theorem]{Definition}
\newtheorem{corollary}[theorem]{Corollary}
\newtheorem{observation}[theorem]{Observation}

\newcommand*\myleft[1]{\negthinspace\left#1}
\newcommand{\Verm}[1]{\mathcal{V}\myleft( #1 \right)}
\newcommand{\mVerm}[1]{\mathcal{V}^{\theta}_{\mathrm{m}}\myleft( #1 \right)}
\newcommand{\aVerm}[1]{\mathcal{V}^{\theta}_{\mathrm{a}}\myleft( #1 \right)}
\newcommand{\oneset}[1]{\left\{\, \mathinner{#1} \,\right\}}
\newcommand{\abs}[1]{\left|\mathinner{#1}\right|}
\newcommand{\IDX}[2]{#1_{[#2]}}               % i-th letter of w
\newcommand{\pPRE}{<_\mathrm{p}}              % proper prefix
\newcommand{\PRE}{\leq_\mathrm{p}}            % prefix
\newcommand{\SUF}{\leq_\mathrm{s}}            % suffix
\newcommand{\N}{\mathbb{N}}
\newcommand{\eps}{\varepsilon}

\newlength{\intervallbracketheighthalf} %Halbe Laenge der Intervallbegrenzung
\newlength{\intervallbracketheightbelow} %Halbe Laenge der Intervallbegrenzung
\newlength{\intervallbracketheightabove} %Halbe Laenge der Intervallbegrenzung
\newlength{\intervalllabeldistance} % Abstand des Intervall-Labels
\newcommand{\intervall}[5][(0,0)]{
	    \draw #1 to node[#5, outer sep=\intervalllabeldistance]{#4\vphantom{$\theta(b)$}} ($#1 + (#2,0)$); % Linie waagerecht
	    \coordinate(#3left)at #1;
	    \coordinate(#3right)at($#1 + (#2,0)$);
	    \draw ($(#3left) - (0,\intervallbracketheighthalf)$) -- ($(#3left) + (0,\intervallbracketheighthalf)$); % linke Grenze%
	    \draw ($(#3right) - (0,\intervallbracketheighthalf)$) -- ($(#3right) + (0,\intervallbracketheighthalf)$); % rechte Grenze%
}

\newcommand{\intleftopen}[5][(0,0)]{
	    \draw #1 to node[#5, outer sep=\intervalllabeldistance]{#4\vphantom{$\theta(b)$}} ($#1 + (#2,0)$); % Linie waagerecht
	    \coordinate(#3left)at #1;
	    \coordinate(#3right)at($#1 + (#2,0)$);
	    \draw ($(#3right) - (0,\intervallbracketheighthalf)$) -- ($(#3right) + (0,\intervallbracketheighthalf)$); % rechte Grenze%
}

\newcommand{\addint}[6][3]{%
	\coordinate (tmpleft) at ($(#2left)!#3!(#2right)$); % Linker Punkt des Teilintervalls
	\coordinate (tmpright) at ($(#2left)!#4!(#2right)$); % Rechter Punkt des Teilintervalls
	\coordinate (tmpmidth) at ($(tmpleft)!0.5!(tmpright)$); % Punkt in der mitte des Teilintervalls
 
 	\ifnum #1=3
 	\draw ($(tmpleft)-(0,\intervallbracketheightbelow)$) -- ($(tmpleft)+(0,\intervallbracketheightabove)$); % Linke Grenze einzeichnen
 	\draw ($(tmpright)-(0,\intervallbracketheightbelow)$) -- ($(tmpright)+(0,\intervallbracketheightabove)$); % Rechte Grenze einzeichnen
 	\fi
 	\ifnum #1=1
 	\draw ($(tmpleft)-(0,\intervallbracketheightbelow)$) -- ($(tmpleft)+(0,\intervallbracketheightabove)$); % Linke Grenze einzeichnen
 	\fi
 	\ifnum #1=2
	\draw ($(tmpright)-(0,\intervallbracketheightbelow)$) -- ($(tmpright)+(0,\intervallbracketheightabove)$); % Rechte Grenze einzeichnen
 	\fi
	\draw (tmpmidth) node[inner sep=0pt, minimum size=0pt, outer sep=\intervalllabeldistance, label=#6:#5\vphantom{$\theta(b)$}]  {}; % Label einzeichnen
}

\newcommand{\addintbelow}[5][3]{%
	\setlength{\intervallbracketheightabove}{0.0cm}
	\addint[#1]{#2}{#3}{#4}{#5}{below}
	\setlength{\intervallbracketheightabove}{\intervallbracketheighthalf}
}
\title{Pattern Avoidability with Involution}
\author{Bastian Bischoff \qquad\qquad  Dirk Nowotka\thanks{This work has been
supported by the DFG Heisenberg grant 582013.}
\institute{Institute for Formal Methods in Computer Science \\
        Universit\"at Stuttgart, Germany}
\email{bastian.bischoff@googlemail.com \qquad\qquad
       nowotka@fmi.uni-stuttgart.de} 
}
\def\titlerunning{Pattern Avoidability with Involution}
\def\authorrunning{B. Bischoff \& D. Nowotka}

\begin{document}

\maketitle

\begin{abstract}
  An infinte word $w$ avoids a pattern $p$
  with the involution $\theta$ if there is no
  substitution for the variables in $p$ and no involution $\theta$
  such that the resulting word is a factor of $w$.
%  An involution $\theta$ is a mapping such that $\theta^2$ is the
%  identity.
  We investigate the avoidance of patterns with respect to the size
  of the alphabet. For example, it is shown that the pattern
  $\alpha\,\theta(\alpha)\,\alpha$ can be avoided over three letters but
  not two letters, whereas it is well known that $\alpha\,\alpha\,\alpha$
  is avoidable over two letters.
\end{abstract}

\section{Introduction}
The avoidability of patterns in infinite words is an old area of interest
with a~first systematic study going back to Thue~\cite{Thue:06,Thue:12}.
This field includes rediscoveries and studies by many authors over
the last one hundred years; see for example~\cite{Currie:05}
and~\cite{Cassaigne:02} for surveys. 
In this article, we are concerned with a variation of the theme by considering
avoidable patterns with involution.
An involution $\theta$ is a mapping such that $\theta^2$ is the identity.
We consider morphic, where $\theta(uv)=\theta(u)\theta(v)$, and antimorphic
involutions, where $\theta(uv)=\theta(v)\theta(u)$.
The subject of this article draws quite some motivation from applications in
biology where the Watson-Crick complement corresponds to an antimorphic
involution in our case.  Our considerations are more general, however, by
considering any alphabet size and also morphic involutions.

During the review phase of this article, James Currie~\cite{Currie:11}
presented a solution
for all those patterns under involution in $\{\alpha,\theta(\alpha)\}^*$
that we do not consider here, which leads to a characterization
of the avoidance index for all unary patterns under involution.

\section{Preliminaries}
Our notation is guided by what is commonly found in the literature, see for
example the first chapter of~\cite{Loth:2} as a reference.
Let $\Sigma$ be a finite alphabet of \emph{letters} and $\Sigma^*$ denote
all \emph{finite} and $\Sigma^\omega$ denote all (right-) \emph{infinite}
words over $\Sigma$. Let $\eps$ denote the empty word. Letters are usually
denoted by $a$, $b$, or $c$, and words over $\Sigma$ are usually denoted by
$u$, $v$, or $w$ in this paper. The $i$-th letter of a word $w$ is denoted
by $w_{[i]}$, that is, $w=w_{[1]}w_{[2]}\cdots w_{[n]}$ if $w$ is finite,
and the length $n$ of~$w$ is denoted by $|w|$ as usual.

%
%A typical pattern is $\alpha \alpha$, the word $w = (ab)^\omega$ contains the pattern $\alpha \alpha$ for $\alpha = ab$. For $\abs{\Sigma}= 2$, the pattern $\alpha \alpha$ is \emph{unavoidable}, that means this pattern appears in every infinitely long word $w \in \Sigma^\omega$. We want to check if patterns, whose variables partially are mapped by a morphic or antimorphic involution, are avoidable or unavoidable. An example for such a pattern is $\alpha \theta(\alpha) \alpha$. Whether a pattern is avoidable or unavoidable often depends on the size of the alphabet $\Sigma$. Many statements about avoidable and unavoidable patterns without involutions can be found in.

Besides $\Sigma$ we need another finite set $E$ of symbols. The elements
of $E$ are called \emph{variables} and we usually denote them by
$\alpha$, $\beta$, or $\gamma$. Words in $E^*$ are called \emph{patterns}.
For example $\alpha \beta \alpha \in E^*$ is a pattern consisting of
the variables $\alpha$ and $\beta$ in $E$. We assign to every pattern
a \emph{pattern language} over the alphabet $\Sigma$. This language contains
every word, that can be generated by substituting all variables in the pattern
by non-empty words in $\Sigma^*$. For example the pattern language of
the pattern $\alpha \alpha$ over $\Sigma = \oneset{a,b}$ is
$\oneset{aa, bb, aaaa, abab, baba, bbbb, \ldots}$.

We say that a word $w$ \emph{avoids} a pattern, if no factor of $w$ exists,
that is in the pattern language. On the other hand, if a factor of $w$ is
an element of the pattern language, we say $w$ \emph{contains} the pattern.
If for a given pattern $e$ and an alphabet $\Sigma$ with $k$ elements a word
$w \in \Sigma^\omega$ exists that avoids $e$, then we say that $e$ is
\emph{$k$-avoidable}. Otherwise we call $e$ \emph{$k$-unavoidable}.
We call $k \in \N$ the \emph{avoidance index $\Verm{e}$} of a pattern
$e \in E^*$, if $e$ is $k$-avoidable and $k$ is minimal.
If no such $k$ exists, we define $\Verm{e} = \infty$.

Let $f\colon \oneset{a,b}^* \rightarrow \oneset{a,b}^*$ with $a \mapsto ab$
and $b \mapsto ba$. The fixpoint $t=\lim_{k\to\infty}f^k(a)$ exists and
is called \emph{Thue--Morse word}. The following result is a~classical one.
\begin{theorem}[\cite{Thue:06,Thue:12}]\label{lem:thuemorse}
  The Thue--Morse word avoids the patterns $\alpha \alpha \alpha$ and
  $\alpha \beta \alpha \beta \alpha$.
\end{theorem}

\section{Patterns with Involution}
For introducing patterns with involution, we extend the set of pattern
variables $E$ by adding $\theta(\alpha)$ for all variables $\alpha \in E$
and some involution $\theta$. For the rest of the article, we will stick to
this definition of $E$. Given a morphic or antimorphic involution, we build
the corresponding pattern language by replacing the variables by non-empty
words and, for variables of the form $\theta(\alpha)$, by applying the
involution after the substitution.

For example, let $\theta$ be the morphic involution with $a \mapsto b$ and
$b \mapsto a$ over $\Sigma = \oneset{a,b}$, and let the pattern be
$\alpha\,\theta(\alpha)$. We get the pattern language
$\oneset{ab, ba, aabb, abba, baab, bbaa, \ldots}$.
Every word in $\oneset{a,b}^\omega\setminus(a^\omega \cup b^\omega)$
contains the pattern $\alpha\,\theta(\alpha)$ for the morphic involution
$\theta$ with $a \mapsto b$ and $b \mapsto a$. 

\begin{observation}
  Let $\theta$ be a morphic or antimorphic involution and not the identity
  or reversal mapping. Then every pattern, that contains variables of
  the $\alpha$ and $\theta(\alpha)$, is avoidable.
\end{observation}
Indeed, since $\theta$ is not the identity or reversal mapping,
a letter $a \in \Sigma$
with $\theta(a) \neq a$ exists. Therefore $w = a^\omega$ avoids every pattern
that includes variables $\alpha$ and $\theta(\alpha)$.

Because of this observation we do not have to examine, if patterns are
avoidable or unavoidable for a given involution. So we now change the point
of view. For a given pattern $e \in E^*$, we either look at all morphic or
all antimorphic involutions $\Sigma^* \rightarrow \Sigma^*$ at the same time.
So, we examine, for example, if an infinite word $w \in \Sigma^\omega$
exists, that avoids a pattern $e$ for all morphic involutions.

\begin{definition}
  Let $e \in E^*$ be a pattern, possibly with variables of the form
  $\theta(\alpha)$. We call $k \in \N$ the morphic (antimorphic)
  $\theta$-avoidance index $\mVerm{e}$ ($\aVerm{e}$) of $e \in E^*$, if
  an infinite word $w \in \Sigma^\omega$ over $\Sigma$ with $\abs{\Sigma}=k$
  exists, that avoids the pattern $e$ for all morphic (antimorphic)
  involutions $\Sigma^* \rightarrow \Sigma^*$ and $k$ is minimal.
  If this doesn't hold for any $k \in \N$, we define $\mVerm{e} = \infty$
  ($\aVerm{e} = \infty$).
\end{definition}

We establish the first facts about avoidance of pattern
$\alpha\,\theta(\alpha)\,\alpha$.

\begin{lemma}\label{lem:ATAA:gr:2:morph}
  Let $\Sigma$ be a binary alphabet. Then there is no word
  $w \in \Sigma^\omega$, that avoids the pattern
  $\alpha\,\theta(\alpha)\,\alpha$ for all morphic involutions
  $\theta\colon\Sigma^* \to \Sigma^*$. That is,
  $\mVerm{\alpha\,\theta(\alpha)\,\alpha}>2$.
\end{lemma}
\begin{proof}
  Let $\Sigma  = \oneset{a,b}$. We try to construct a word
  $w \in \Sigma^\omega$, that avoids $e=\alpha\,\theta(\alpha)\,\alpha$
  for all morphic involutions and bring this to a contradiction. For example,
  this word must not contain $aaa$, $bbb$, $aba$ or $bab$ as a factor.
  Without loss of generality $w$ begins with $a$.
  
  Case $1$: Assumed the word $w$ begins with $ab$. Then this prefix must
  be followed by $b$, $abb \pPRE w$. The next letter must be an $a$, the
  fifth must be an $a$ too. So we have $abbaa \pPRE w$. If the following
  letter is an $a$, $aaa$ is a factor of $w$. So the next letter must be
  the letter $b$. But for the morphic involution $\theta$ with $a \mapsto b$
  and $b \mapsto a$ the word $ab \theta(ab) ab$ is a factor of $w$.
  
  Case $2$: The argument for the case $aa \PRE w$ is analogous to case 1.
\end{proof}

The proof of the following lemma is analogous to the previous one.
\begin{lemma}\label{lem:ATAA:gr:2:antimorph}%8.6
  Let $\Sigma$ be a binary alphabet.
  There is no word $w \in \Sigma^\omega$, that avoids the pattern
  $\alpha\,\theta(\alpha)\,\alpha$ for all antimorphic involutions
  $\theta\colon\Sigma^* \to \Sigma^*$. That is,
  $\aVerm{\alpha\,\theta(\alpha)\,\alpha}>2$.
\end{lemma}

\section{Main Result}
In this section, we establish the $\theta$-avoidance indices for
the pattern $\alpha\,\theta(\alpha)\,\alpha$ in the morphic and
antimorphic case. We start with the morphic case.

\begin{theorem}\label{thm:ATA3}
  It holds that $\mVerm{\alpha \theta(\alpha) \alpha}=3$.
\end{theorem}
\begin{proof}
  Let $\Sigma$ an alphabet with three elements, $\Sigma = \oneset{a,b,c}$.
  Let $v$ be the infinitely long Thue--Morse word over the letters $a'$
  and $b'$. Furthermore let $w \in \Sigma^\omega$ be the word, that is the outcome
  of replacing every $a'$ in $v$ by $aacb$ and $b'$ by $accb$. We will show,
  that $w$ avoids the pattern $\alpha \theta(\alpha) \alpha$ for all morphic
  involutions. For better readability, we define $x = aacb$ and $y = accb$.

  We assume it exists a morphic involution $\theta$ and a substitution
  for $\alpha$, such that $\alpha \theta(\alpha) \alpha$ is a factor of $w$.
  Proof by contradiction. First, we examine the possibilities of replacing
  the variable $\alpha$ by words $u \in \Sigma^+$  of length $\abs{u}<7$.
  The word $u\,\theta(u)\,u$ has a maximal length of $18$. Therefore there must
  exist a morphic involution so that $u\,\theta(u)\,u$ is a factor of a word
  $w' \in \oneset{x,y}^6$. Because of Theorem~\ref{lem:thuemorse}, the words
  $xxx$, $yyy$, $xyxyx$ and $yxyxy$ can not be a factor of $w'$. A computer
  program can easily check these finite possibilities with the result,
  that no words $u$ and $w'$ exist, which fulfill the conditions. 
  Now we assume $\alpha$ gets replaced by a word $u \in \Sigma^+$ with
  $\abs{u} \geq 7$. Then, the word $u$ contains $aacb$ or $accb$. Without
  loss of generality, $u$ contains $aacb$. Therefore, $\theta(u)$ contains
  the factor $\theta(aac) = \theta(a)\,\theta(a)\,\theta(c)$.
  In addition $\theta(u)$ and for this reason $\theta(a)\,\theta(a)\,\theta(c)$
  is a factor of $w$. There are only two possibilities for two succeeding
  identical letters in $w$. Either these letters are two letters $c$ followed
  by the letter $b$, or two letters $a$ are followed by the letter $c$.
  This implies, that $u\,\theta(u)\,u$ can only be a~factor of $w$,
  if $\theta$ is the identity mapping. Furthermore this implies
  $\abs{u} = 4 \cdot k$ for a $k \in \N$. This is visualized in
  Fig.~\ref{fig:ATA3}, where $w_i, w_{i'}, w_{i''} \in \oneset{x,y}$ holds
  for all $0 \leq i \leq k$. If the word
  $\IDX{\myleft(w_{0}\right)}{2}\IDX{\myleft(w_{0}\right)}{3}
  \IDX{\myleft(w_{0}\right)}{4}$ or $\IDX{\myleft(w_{0}\right)}{1}
  \IDX{\myleft(w_{0}\right)}{2}\IDX{\myleft(w_{0}\right)}{3}
  \IDX{\myleft(w_{0}\right)}{4}=w_0$ is a prefix of the first $u$
  in Fig.~\ref{fig:ATA3}, then the following equations apply:
  \newlength{\tmpArrayLength}
  \settowidth{\tmpArrayLength}{$w_{k-1''}$}
  \begin{center}
  \begin{tabular}{C{\tmpArrayLength}cC{\tmpArrayLength}cC{\tmpArrayLength}}
    $w_0$ & $=$ & $w_{0'}$ & $=$ & $w_{0''}$ \tabularnewline
    $w_1$ & $=$ & $w_{1'}$ & $=$ & $w_{1''}$ \tabularnewline
    $\vdots$  & & $\vdots$ & & $\vdots$ \tabularnewline
    $w_{k-1}$ & $=$ & $w_{k-1'}$ & $=$ & $w_{k-1''}$ 
  \end{tabular}
  \end{center} 
  The word $w_0 w_1 \ldots w_{k-1} \, w_{0'} w_{1'} \ldots w_{k-1'} \, w_{0''}
  w_{1''} \ldots w_{k-1''} = \myleft( w_0 w_1 \ldots w_{k-1} \right)^3$
  is a factor of $w$. Because of $w_i \in \oneset{x,y} $ for all
  $0 \leq i \leq k-1$, this is a contradiction to Lemma~\ref{lem:thuemorse}.
  On the other hand, if only $\IDX{\myleft(w_{0}\right)}{3}
  \IDX{\myleft(w_{0}\right)}{4}$ or $\IDX{\myleft(w_{0}\right)}{4}$
  is a prefix of $u$, then $w_0 \neq w_{0'}$ is possible. But in this case
  $\IDX{\myleft(w_{k''}\right)}{1} \IDX{\myleft(w_{k''}\right)}{2}$ or
  $\IDX{\myleft(w_{k''}\right)}{1} \IDX{\myleft(w_{k''}\right)}{2}
  \IDX{\myleft(w_{k''}\right)}{3}$ is a suffix of the third $u$.
  This implies
  \begin{center}
  \begin{tabular}{C{\tmpArrayLength}cC{\tmpArrayLength}cC{\tmpArrayLength}}
    $w_1$ & $=$ & $w_{1'}$ & $=$ & $w_{1''}$ \tabularnewline
    $w_2$ & $=$ & $w_{2'}$ & $=$ & $w_{2''}$ \tabularnewline
    $\vdots$  & & $\vdots$ & & $\vdots$ \tabularnewline
    $w_{k}$ & $=$ & $w_{k'}$ & $=$ & $w_{k''}$ 
  \end{tabular}
  \end{center} 
  and $w_1 w_2 \ldots w_{k} \, w_{1'} w_{2'} \ldots w_{k'} \, w_{1''}
  w_{2''} \ldots w_{k''} = \myleft( w_1 w_2 \ldots w_{k} \right)^3$
  is a factor of $w$. Again, this is a contradiction
  to Lemma~\ref{lem:thuemorse}. The theorem follows with
  Lemma~\ref{lem:ATAA:gr:2:morph}.
\end{proof}

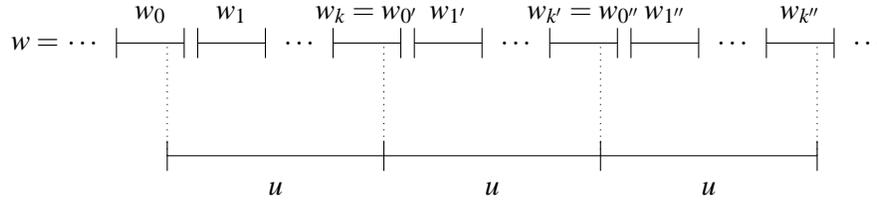
\begin{figure}
  \centering
    \begin{tikzpicture}[xscale=0.9]
    \draw (-1.2, 0) node {$w =$};
    \draw (-0.5, 0) node {$\ldots$};
    \intervall[(0,0)]{1}{w0}{$w_0$}{above}
    \intervall[(1.2,0)]{1}{w1}{$w_1$}{above}
    \draw (2.7,0) node {$\ldots$};
    \intervall[(3.2,0)]{1}{wk}{$w_k = w_{0'}$}{above}
    \intervall[(4.4,0)]{1}{w1a}{$w_{1'}$}{above}
    \draw (5.9,0) node {$\ldots$};
    \intervall[(6.4,0)]{1}{wka}{$w_{k'} = w_{0''}$}{above}
    \intervall[(7.6,0)]{1}{w1b}{$w_{1''}$}{above}
    \draw (9.1,0) node {$\ldots$};
    \intervall[(9.6,0)]{1}{wkb}{$w_{k''}$}{above}
    \draw (11.1, 0) node {$\ldots$};

    \intervall[(0.75,-1.5)]{3.2}{u1}{$u$}{below}
    \intleftopen[(3.95,-1.5)]{3.2}{u2}{$u$}{below}
    \intleftopen[(7.15,-1.5)]{3.2}{u3}{$u$}{below}
    \draw[dotted] (0.75, -1.5) -- (0.75,0);
    \draw[dotted] (3.95, -1.5) -- (3.95,0);
    \draw[dotted] (7.15, -1.5) -- (7.15,0);
    \draw[dotted] (10.35, -1.5) -- (10.35,0);
    \end{tikzpicture}
    \caption{Part of $w$ to illustrate the factor $uuu$}
    \label{fig:ATA3}
\end{figure}

The result of Theorem~\ref{thm:ATA3} transfers also to the antimorphic case.
\begin{theorem}\label{thm:ATA3:antimorph}
  It holds that $\aVerm{\alpha \theta(\alpha) \alpha}=3$.
\end{theorem}
\begin{proof}
  This proof follows the proof of the previous theorem. Let $\Sigma$ be
  an alphabet with three elements, $\Sigma = \oneset{a,b,c}$. Further,
  let $v$ be the Thue-Morse word over the letters $a'$ and $b'$.
  Let $w \in \Sigma^\omega$ be the word, that we get by replacing $a'$ in
  $v$ by $aabbc$ and $b'$ by $aaccb$. We will show, that $w$ avoids the
  pattern $\alpha\,\theta(\alpha)\,\alpha$ for all antimorphic involutions.
  For better readability, we define $x = aabbc$ and $y = aaccb$.

  We assume that there exists an antimorphic involution and a substitution
  of $\alpha$ by a word $u \in \Sigma^+$ in such a way, that $u\,\theta(u)\,u$ is
  a factor of $w$. First we suppose that $\abs{u}<9$ holds. The word
  $u\,\theta(u)\,u$ then has a maximal length of $24$ and $u\,\theta(u)\,u$
  is factor of a word $w' \in \oneset{x,y}^6$. The word $xxx$, $yyy$,
  $xyxyx$, and $yxyxy$ must not be a factor of $w'$ because of
  Lemma~\ref{lem:thuemorse}. A computer program can check these finite
  possibilities with the result, that no words $u$ and $w'$ exist that 
  fulfill these conditions for an antimorphic involution $\theta$.
  So ,$\abs{u} \geq 9$ must hold and $u$ contains at least one word $x$
  or $y$ completely. We now look at the first $u$ of the factor
  $u\,\theta(u)\,u$ of $w$. Let $w_1 w_2' \SUF u$ with
  $w_1, w_2 \in \oneset{x,y}$, $w_2 = w_2' w_2 ''$ and $\abs{w_2'} < 5$.
  We get Fig.~\ref{fig:ATA3:antimorph} where $w_3, w_4 \in \oneset{x,y}$.
  Without loss of generality, let $w_1 = x = aabbc$. Then $\theta(u)$
  and therefore $w_2 w_3 w_4$ contains the word
  $\theta(aabbc) = \theta(c)\,\theta(b)\,\theta(b)\,\theta(a)\,\theta(a)$
  with length $5$ as a factor. Hence we look at the following words:
  \begin{align*}
    xx &= aabbc \, aabbc\\
    xy &= aabbc \, aaccb\\
    yx &= aaccb \, aabbc\\
    yy &= aaccb \, aaccb\ .
  \end{align*}
  Only $xx$ contains $\theta(c)\,\theta(b)\,\theta(b)\,\theta(a)\,\theta(a)$
  for the antimorphic involution $\theta$ with $a \mapsto b$, $b \mapsto a$,
  and $c \mapsto c$. Because of $w_1 = x$, the equation $w_2 w_3 = xx$ is
  a contradiction to Lemma~\ref{lem:thuemorse}. The case $w_2 w_3 w_4= y x x$
  remains. Now there are five possibilities for the position of $u$,
  see Fig.~\ref{fig:ATA3:antimorph2}. It is easy to check, that in all five
  cases $\theta(u) \PRE w_2'' w_3 w_4$ respectively
  $w_2'' w_3 w_4 \PRE \theta(u)$ doesn't hold. So our assumption, that there
  exists an antimorphic involution $\theta$ and a word $u \in \Sigma^+$ with
  $u\,\theta(u)\,u$ is a factor of $w$, was wrong. The theorem follows with
  Lemma~\ref{lem:ATAA:gr:2:antimorph}. 
\end{proof}

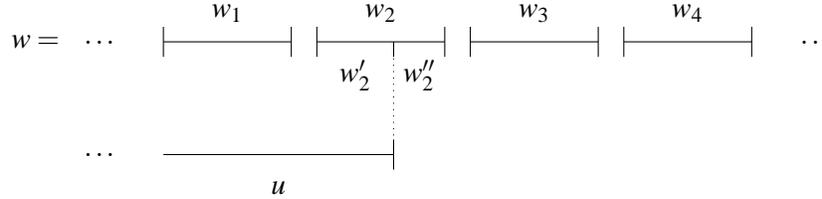
\begin{figure}
  \centering
  \begin{tikzpicture}[xscale=1.7]
    \draw (-1.0, 0) node {$w =$};
    \draw (-0.5, 0) node {$\ldots$};
    \intervall[(0,0)]{1}{w1}{$w_1$}{above}
    \intervall[(1.2,0)]{1}{w2}{$w_2$}{above}
    \addintbelow[2]{w2}{0.0}{0.6}{$w_2'$}
    \addintbelow[0]{w2}{0.6}{1.0}{$w_2''$}
    \intervall[(2.4,0)]{1}{w3}{$w_3$}{above}
    \intervall[(3.6,0)]{1}{w4}{$w_4$}{above}
    \draw (5.1, 0) node {$\ldots$};
    \draw (-0.5, -1.5) node {$\ldots$};
    \intleftopen[(0,-1.5)]{1.8}{u1}{$u$}{below}
    \draw[dotted] (1.8, -1.5) -- (1.8,-\intervallbracketheightbelow);
  \end{tikzpicture}
  \caption{Part of $w$ and the factor $u$ of $w$}
  \label{fig:ATA3:antimorph}
\end{figure}

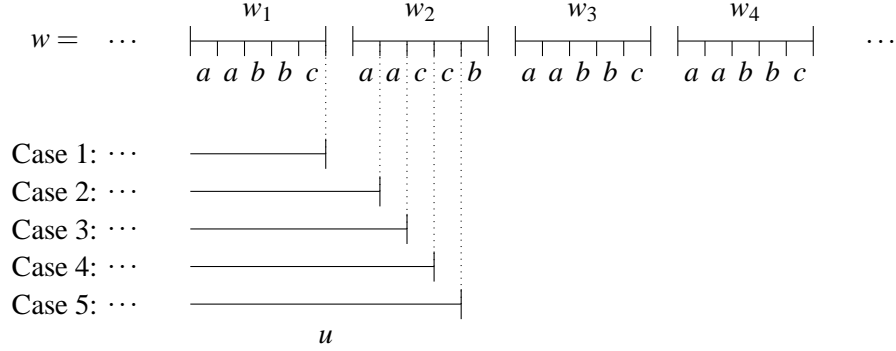
\begin{figure}
  \centering
  \begin{tikzpicture}[xscale=1.8]
    \draw (-1.0, 0) node {$w =$};
    \draw (-0.5, 0) node {$\ldots$};
    \intervall[(0,0)]{1}{w1}{$w_1$}{above}
    \addintbelow[2]{w1}{0.0}{0.2}{$a$}
    \addintbelow[2]{w1}{0.2}{0.4}{$a$}
    \addintbelow[2]{w1}{0.4}{0.6}{$b$}
    \addintbelow[2]{w1}{0.6}{0.8}{$b$}
    \addintbelow[0]{w1}{0.8}{1.0}{$c$}

    \intervall[(1.2,0)]{1}{w2}{$w_2$}{above}
    \addintbelow[2]{w2}{0.0}{0.2}{$a$}
    \addintbelow[2]{w2}{0.2}{0.4}{$a$}
    \addintbelow[2]{w2}{0.4}{0.6}{$c$}
    \addintbelow[2]{w2}{0.6}{0.8}{$c$}
    \addintbelow[0]{w2}{0.8}{1.0}{$b$}

    \intervall[(2.4,0)]{1}{w3}{$w_3$}{above}
    \addintbelow[2]{w3}{0.0}{0.2}{$a$}
    \addintbelow[2]{w3}{0.2}{0.4}{$a$}
    \addintbelow[2]{w3}{0.4}{0.6}{$b$}
    \addintbelow[2]{w3}{0.6}{0.8}{$b$}
    \addintbelow[0]{w3}{0.8}{1.0}{$c$}

    \intervall[(3.6,0)]{1}{w4}{$w_4$}{above}
    \addintbelow[2]{w4}{0.0}{0.2}{$a$}
    \addintbelow[2]{w4}{0.2}{0.4}{$a$}
    \addintbelow[2]{w4}{0.4}{0.6}{$b$}
    \addintbelow[2]{w4}{0.6}{0.8}{$b$}
    \addintbelow[0]{w4}{0.8}{1.0}{$c$}

    \draw (5.1, 0) node {$\ldots$};

    \draw (-1, -1.5) node {Case $1$:};
    \draw (-0.5, -1.5) node {$\ldots$};
    \intleftopen[(0,-1.5)]{1.0}{u1}{}{above}
    \draw[dotted] (1.0, -1.5) -- (1.0,0);

    \draw (-1, -2.0) node {Case $2$:};
    \draw (-0.5, -2.0) node {$\ldots$};
    \intleftopen[(0,-2.0)]{1.4}{u2}{}{below}
    \draw[dotted] (1.4, -2.0) -- (1.4,0);

    \draw (-1, -2.5) node {Case $3$:};
    \draw (-0.5, -2.5) node {$\ldots$};
    \intleftopen[(0,-2.5)]{1.6}{u3}{}{below}
    \draw[dotted] (1.6, -2.5) -- (1.6,0);

    \draw (-1, -3.0) node {Case $4$:};
    \draw (-0.5, -3.0) node {$\ldots$};
    \intleftopen[(0,-3.0)]{1.8}{u4}{}{below}
    \draw[dotted] (1.8, -3.0) -- (1.8,0);

    \draw (-1, -3.5) node {Case $5$:};
    \draw (-0.5, -3.5) node {$\ldots$};
    \intleftopen[(0,-3.5)]{2.0}{u5}{$u$}{below}
    \draw[dotted] (2.0, -3.5) -- (2.0,0);

  \end{tikzpicture}
  \caption{Illustration of possible positions of the factor $u$ of $w$}
  \label{fig:ATA3:antimorph2}%
\end{figure}

\section{Complementary Patterns}
In this section, patterns similar to $\alpha\,\theta(\alpha)\,\alpha$
are considered.

For the next lemma we need a further definition. Let $e \in E^*$ be a pattern
consisting of variables of the form $\alpha$ and $\theta(\alpha)$ and $e'$
be the pattern that we get, when all variables $\alpha$ and $\theta(\alpha)$
in $e$ are switched. We call $e' \in E$ the \emph{$\theta$-complementary}
pattern of $e$. For example the $\theta$-complementary pattern
of $\alpha\,\alpha\,\theta(\alpha)\,\beta$ is
$\theta(\alpha)\,\theta(\alpha)\,\alpha\,\theta(\beta)$. For this definition
it doesn't matter if morphic or antimorphic involutions are examined.

\begin{lemma}
  Let $e \in E^*$ be a pattern and $e' \in E$ be the $\theta$-complementary
  pattern of $e$. Then $\aVerm{e} = \aVerm{e'}$ and $\mVerm{e} = \mVerm{e'}$.
\end{lemma}
\begin{proof}
  First of all we show $\mVerm{e} = \mVerm{e'}$. For better readability,
  we replace the variable $\alpha$ in the pattern $e'$ by $\alpha'$
  and $\theta(\alpha)$ by $\theta(\alpha')$. We assume
  a word $w \in \Sigma^\omega$ contains the pattern $e$ for a morphic
  involution and a substitution of $\alpha$ by $u \in \Sigma^+$.
  Then $w$ contains the pattern $e'$ for the same morphic involution
  by substituting $\alpha'$ by $\theta(u)$. Symmetry reasons imply:
  \begin{align*}
    & \text{It exists a morphic involution } \theta \text{ so that } w
      \text{ contains the pattern } e \text{.}\\
    \Leftrightarrow \; \; & \text{It exists a morphic involution }
      \theta' \text{ so that } w \text{ contains the pattern } e' \text{.}
  \end{align*}
  By negation we get:
  \begin{align*}
    &\text{The word }w\in\Sigma^\omega \text{ avoids the pattern } e\text{.}\\
    \Leftrightarrow \; \; & \text{The word }w\in\Sigma^\omega
    \text{ avoids the pattern } e'\text{.}
  \end{align*}
  The equation $\mVerm{e} = \mVerm{e'}$ follows. The proof of
  $\aVerm{e} = \aVerm{e'}$ is identical.
\end{proof}

Note the following $\theta$-free patterns; see~\cite{Cassaigne:02}.
\begin{observation}\label{not:vermeidbar:lothaire}
  The patterns $\alpha\alpha$, $\alpha\alpha\beta$,
  $\beta\alpha\alpha$, $\alpha\alpha\beta\alpha$,
  $\alpha\beta\beta\alpha$, $\alpha\alpha\beta\beta$,
  $\alpha\beta\alpha\beta$, $\alpha\alpha\beta\alpha\alpha$,
  and $\alpha\alpha\beta\alpha\beta$ are $2$-unavoidable and $3$-avoidable.
\end{observation}

\begin{lemma}\label{lem:vermeid:index}%8.4
  Let $e \in E^*$ be a pattern, that contains the variables $\alpha$
  and $\theta(\alpha)$. Further, $e$ contains no other variable of the
  form $\theta(\gamma)$. Let $e'$ be the pattern when all occurrences
  of $\theta(\alpha)$ in $e$ are replaced by $\alpha$. The pattern $e''$
  obtained when all occurrences of $\theta(\alpha)$ in $e$ are replaced by
  a new variable $\beta$.
  
  Then $\Verm{e'} \leq \mVerm{e} \leq \Verm{e''}$ and
  $\aVerm{e} \leq \Verm{e''}$.
\end{lemma}
\begin{proof}
  The relation $\Verm{e'} \leq \mVerm{e}$ holds, since the morphic
  $\theta$-avoidance index considers all morphic involutions, including
  the identity mapping. Now say $\Verm{e''} = k$, i.e., a~word
  $w \in \Sigma^\omega$ exists, that avoids the pattern $e''$.Then this word
  also avoids the pattern $e$ for all morphic and antimorphic involutions.
  Therefore the relations $\mVerm{e} \leq \Verm{e''}$ and
  $\aVerm{e} \leq \Verm{e''}$ hold.
\end{proof}

\begin{lemma}\label{lem:AATA:3}
 It holds that
 $\aVerm{\alpha\,\alpha\,\theta(\alpha)} =
  \mVerm{\alpha\,\alpha\,\theta(\alpha)} = 3$.
\end{lemma}
\begin{proof}
  According to Observation~\ref{not:vermeidbar:lothaire} the equation
  $\Verm{\alpha\,\alpha\,\beta} = 3$ holds. Lemma~\ref{lem:vermeid:index}
  implies $\aVerm{\alpha\,\alpha\,\theta(\alpha)}$,
  $\mVerm{\alpha\,\alpha\,\theta(\alpha)} \leq 3$.
  We show by contradiction, that it holds that
  $\aVerm{\alpha\,\alpha\,\theta(\alpha)} \neq 2$. The proof for
  the relation $\mVerm{\alpha\,\alpha\,\theta(\alpha)} \neq 2$ is analogous.
  Assuming a word $w \in \Sigma^\omega$ with $\Sigma = \oneset{a,b}$ exists
  that avoids the pattern $\alpha\,\alpha\,\theta(\alpha)$ for all antimorphic
  involutions. Then $w$ contains neither $aa$ nor $bb$ as a factor. Without
  loss of generality $w$ begins with the letter $a$. It follows
  that $w = (ab)^\omega$. But $w = (ab)^\omega$ contains the pattern
  $\alpha\,\alpha\,\theta(\alpha)$ for $\alpha=ab$ and the antimorphic
  involution defined by $a \mapsto b$ and $b \mapsto a$. This is
  a contradiction to our assumption. Therefore
  $\aVerm{\alpha\,\alpha\,\theta(\alpha)} \neq 2$ holds and analogously
  $\mVerm{\alpha\,\alpha\,\theta(\alpha)} \neq 2$. We get
  $\aVerm{\alpha\,\alpha\,\theta(\alpha)} =
  \mVerm{\alpha\,\alpha\,\theta(\alpha)} = 3$.
\end{proof}

\begin{lemma}\label{lem:TAAA:3}
  It holds that $\aVerm{\theta(\alpha)\,\alpha\,\alpha} =
                 \mVerm{\theta(\alpha)\,\alpha\,\alpha} = 3$.
\end{lemma}
\begin{proof}
  The proof is analogous to the proof of Lemma~\ref{lem:AATA:3}.
\end{proof}

\begin{corollary}
  $ $
  \begin{enumerate}
    \item $\mVerm{\theta(\alpha)\,\alpha\,\theta(\alpha)} =
             \aVerm{\theta(\alpha)\,\alpha\,\theta(\alpha)} = 3$
      by Theorem~\ref{thm:ATA3} and~\ref{thm:ATA3:antimorph}.
    \item $\mVerm{\theta(\alpha)\,\theta(\alpha)\,\alpha} =
           \aVerm{\theta(\alpha)\,\theta(\alpha)\,\alpha} = 3$
      by Lemma~\ref{lem:AATA:3}.
    \item $\mVerm{\alpha\,\theta(\alpha)\,\theta(\alpha)} =
           \aVerm{\alpha\,\theta(\alpha)\,\theta(\alpha)} = 3$
      by Lemma~\ref{lem:TAAA:3}.
  \end{enumerate}
\end{corollary}

\section*{Acknowledgement}
We would like to thank the anonymous referees for their careful reviews and
constructive comments.

\bibliography{BibtexDatabase}

\begin{thebibliography}{1}
\providecommand{\bibitemdeclare}[2]{}
\providecommand{\urlprefix}{Available at }
\providecommand{\url}[1]{\texttt{#1}}
\providecommand{\href}[2]{\texttt{#2}}
\providecommand{\urlalt}[2]{\href{#1}{#2}}
\providecommand{\doi}[1]{doi:\urlalt{http://dx.doi.org/#1}{#1}}
\providecommand{\bibinfo}[2]{#2}

\bibitemdeclare{inbook}{Cassaigne:02}
\bibitem{Cassaigne:02}
\bibinfo{author}{J.~Cassaigne} (\bibinfo{year}{2002}):
  \emph{\bibinfo{title}{Unavoidable Patterns}}, chapter~\bibinfo{chapter}{3},
  pp. \bibinfo{pages}{111--134}.
\newblock In  \cite{Loth:2}.

\bibitemdeclare{article}{Currie:05}
\bibitem{Currie:05}
\bibinfo{author}{J.~Currie} (\bibinfo{year}{2005}):
  \emph{\bibinfo{title}{Pattern avoidance: themes and variations}}.
\newblock {\sl \bibinfo{journal}{Theoret. Comput. Sci.}}
  \bibinfo{volume}{339}(\bibinfo{number}{1}), pp. \bibinfo{pages}{7--18},
  \doi{10.1016/j.tcs.2005.01.004}.

\bibitemdeclare{article}{Currie:11}
\bibitem{Currie:11}
\bibinfo{author}{J.~Currie} (\bibinfo{year}{2011}):
  \emph{\bibinfo{title}{Pattern avoidance with involution}}.
\newblock {\sl \bibinfo{journal}{CoRR}} \bibinfo{volume}{abs/1105.2849}.
\newblock \urlprefix\url{http://arxiv.org/abs/1105.2849}.

\bibitemdeclare{book}{Loth:2}
\bibitem{Loth:2}
\bibinfo{author}{M.~Lothaire} (\bibinfo{year}{2002}):
  \emph{\bibinfo{title}{Algebraic Combinatorics on~Words}}.
\newblock \bibinfo{publisher}{Cambridge University Press},
  \bibinfo{address}{Cambridge, UK}.

\bibitemdeclare{article}{Thue:06}
\bibitem{Thue:06}
\bibinfo{author}{A.~Thue} (\bibinfo{year}{1906}):
  \emph{\bibinfo{title}{{\"U}ber unendliche {Z}eichenreihen}}.
\newblock {\sl \bibinfo{journal}{Norske Vid. Skrifter I.\,Mat.-Nat.\,Kl.,
  Christiania}} \bibinfo{volume}{7}, pp. \bibinfo{pages}{1--22}.

\bibitemdeclare{article}{Thue:12}
\bibitem{Thue:12}
\bibinfo{author}{A.~Thue} (\bibinfo{year}{1912}):
  \emph{\bibinfo{title}{{\"U}ber die gegenseitige {L}age gleicher {T}eile
  gewisser {Z}eichenreihen}}.
\newblock {\sl \bibinfo{journal}{Norske Vid. Skrifter I.\,Mat.-Nat.\,Kl.,
  Christiania}} \bibinfo{volume}{1}, pp. \bibinfo{pages}{1--67}.

\end{thebibliography}
\bibliographystyle{eptcs}

\end{document}